\documentclass[graybox]{svmult}
\usepackage{tikz}

\usepackage{type1cm}        
\usepackage{makeidx}         
\usepackage{graphicx}        
\usepackage{multicol}        
\usepackage[bottom]{footmisc}

\usepackage{newtxtext}       %
\usepackage{newtxmath}       
\usepackage{threeparttable}
\usepackage{tabularx}

\makeindex             
\usepackage{algorithm}
\usepackage[algo2e]{algorithm2e} 
\theoremstyle{plain} 
\spnewtheorem{assumption}{Assumption}{\bfseries}{\itshape}
\spnewtheorem{myproblem}{Problem}{\bfseries}{\itshape}

\begin{document}

\title*{Minimizing cumulative infections in SIS epidemic models over networks via an edge deletion algorithm}
\titlerunning{Minimizing cumulative infections in SIS models over networks via EDA}
\author{Hoang Phi Dung$^*$ and Duong Khanh Ly}
\institute{Hoang Phi Dung$^*$ \at Faculty of Fundamental Sciences - Lab Blockchain, Posts and Telecommunications Institute of Technology, Hanoi, Vietnam, \email{dunghp@ptit.edu.vn} \and Duong Khanh Ly \at Faculty of Information Technology, Posts and Telecommunications Institute of Technology, Hanoi, Vietnam, \email{LyDK.B24CE169@stu.ptit.edu.vn}
\and $^*$corresponding author.}

\maketitle

\abstract{In this paper, we investigate the discrete SIS (Susceptible-Infected-Susceptible) models. We focus on minimizing epidemic spreading over networks by extending an existing edge deletion algorithm to the SIS model. To achieve this, we employ the mean-field approximation to linearize the network dynamics into a deterministic SIS model. We analytically demonstrate that the total number of infections is upper-bounded by a supermodular function, thereby ensuring the efficiency of the edge-deletion approach. To evaluate the proposed method, we conduct experiments on synthetic Erdős–Rényi networks and the real-world dataset collected from Haslemere, England. Numerical simulations validate our theoretical results, confirming that both configurations converge to the stable, disease-free equilibrium.
}

\section{Introduction}

Network-based disease propagation models have now become a sharp tool for analyzing phenomena and simulating information diffusion processes on complex networks \cite{Kempe2003, Acemoglu2013, Chen2014, Pastor-Satorras2015, Shakarian2015, Nowzari2016, Silva2016}. Originating from the biological epidemic model of Kermack and McKendrick \cite{Kermack1927}, the SIR model and its variants are now effectively used to analyze, investigate, and evaluate the spread of information on complex networks. These models are applied in many fields, such as social networks \cite{Chen2014, Shakarian2015}, malware propagation \cite{Mieghem2009,Liu2022}, control systems \cite{Ahn2013, Pare2020}, and others. In the systems and control community or social networks, one of the main concerns is SIS epidemic models on non-trivial networks. These models have been proposed for discrete time \cite{Wang2003, Ahn2013, Pare2020, Gracy2021} and continuous time \cite{Ahn2013,Pare2020}. Research on the spread of information on platforms such as social networks, epidemiology, or malware propagation plays an essential role in communication, information security, and social sciences \cite{Chen2014, Kempe2015, Pham2019, Shakarian2015}. The influence minimization problem and the stability analysis of the models \cite{Kimura2008, Shakarian2015, Pham2019, Liu2022, Xie2023, Cho2024} are the central issues of controlling the spread of malicious information or malware. Recently, there have been very profound investigations into the infection volume in SIS, SIR models \cite{Mieghem2009, Nowzari2016, Pare2020, Yi2022, Youssef2011} and their modified models as SIRS, SDIR \cite{Ruhi2015, Khanh2026} on networks \cite{Mieghem2009, Nowzari2016, Pare2020, Yi2022, Youssef2011, Ruhi2015, Khanh2026}, see also the survey articles such as \cite{Pastor-Satorras2015, Nowzari2016}. Recently, there has been a trend in which many authors investigate an individual-based approach for the models \cite{Youssef2011, Cho2024, Khanh2026}.

In this paper, we consider the Markov chain SIS propagation model similar to the SIS models considered in \cite{Wang2003, Ahn2013, Pare2020, Gracy2021}. Instead of using common infection rates for each vertex as in \cite{Pare2020, Gracy2021}, we consider $\beta_{ij}(t)$ as random variables representing the event that node $i$ is infected from its infected neighbor $j$.

\vspace{1em}
\noindent\textit{Contributions}
 
The main problem is minimizing the scale of the epidemic within the framework of the Markov chain SIS propagation model. The main difficulty is the number of infections in the SIS models (with $2^N$ states) over networks is non-monotone over time. To solve the edge-deletion problem for mitigation under such uncertain conditions, we establish assumptions for the parameters to derive the sufficient conditions for the spectrum of the state transition matrix to be less than $1$ (Theorem \ref{convergence}). On this condition, we propose a monotone supermodular upper bound objective function (Theorem \ref{upper bound}), thereby transforming the NP-hard problem into an effectively solvable form. Furthermore, we adapt an existing edge deletion algorithm (Algorithm \ref{GA}) for the SIS model to identify the optimal set of edges to be removed. However, due to the non-monotone nature of the SIS model, the increased infections at each time step are difficult to determine precisely. Therefore, we propose a new cumulative quantity to apply a supermodular function to the edge-deletion problem. The theoretical findings and algorithmic performance are validated through numerical simulations on both synthetic networks (the Erdos-Renyi Random networks) and a real-world network (the Haslemere dataset). The results demonstrate that our proposed method effectively reduces the infections (Figure \ref{exp-result}).

\vspace{1em}
\noindent\textit{Outline}

The layout of this paper is as follows. Section \ref{sec2} describes the proposed model and defines the core problem, while also providing the conditions regarding the rapid convergence of infection volume. In Section \ref{sec3}, we construct a supermodular objective function and propose an edge-deletion algorithm. Then, we validate the proposed algorithm through data simulations in Section \ref{sec4}. Finally, we summarize the paper in Section \ref{conclusion}.

\section{The information propagation with the SIS model over complex networks}\label{sec2}
\subsection{Model description}
\begin{center}
\begin{tikzpicture}
\node[draw, circle, minimum size=1cm, fill=green] at (-4,0) (S) {$S_i$};
\node[draw, circle, minimum size=1cm,fill=red] at (0,1) (I) {$I_i$};

  \draw[->] (I) -- (S) node[midway, below] {$\delta_i(t)$};

  \node[draw, circle, minimum size=0.5cm,fill=yellow] (I1) at (3,2) {$I_{j_1}$};
  \node[draw, circle, minimum size=0.5cm,fill=yellow] (I2) at (5,1) {$I_{j_2}$};
  \node[draw, circle, minimum size=0.5cm, fill=yellow] (I3) at (5,-1) {$I_{j_3}$};
   \node (Ndots) at (2,-2) {$\cdots$};

  \draw[->] (I1) -- (I) node[pos=0.4, above] {$\beta_{ij_1}(t)$};
  \draw[->] (I2) -- (I) node[pos=0.4, above] {$\beta_{ij_2}(t)$};
  \draw[->] (I3) -- (I) node[pos=0.5, below] {$\beta_{ij_3}(t)$};
  \draw[->] (Ndots) -- (I);
\end{tikzpicture}
\end{center}

Given an undirected graph $G=(V,E)$ with $N$ vertices. At time step $t$, each node $i\in V$ takes one of two states, Susceptible or Infected, corresponding to the variables $S_i(t), I_i(t)$. These variables take the value of $1$ if the node is in the respective state at time $t$ and 0 otherwise. Consequently, $S_i(t)+I_i(t)=1$. The edge connecting a neighboring node $j$ to node $i$ is associated with a random infection variable $\beta_{ij}(t)$ at time $t$. Furthermore, $\delta_i(t)$ represents the random variable for node $i$ reverting to the Susceptible state. Assume that $\beta_{ij}(t)$ and $\delta_i(t)$ are both i.i.d. We denote $[N]:=\{1, \dots, N\}$.  
\begin{assumption}\label{assp1}
    $\beta_{ij}(t)\geq0$ and $\delta_i(t)\geq0$, $\forall i\in[N],\forall t \geq 0$.
\end{assumption}
The recurrence relation characterizing the relationship between these two states can be formulated as 
\begin{align*}
    &I_i(t+1)=(1-h\delta_i(t))I_i(t)+h(1-I_i(t))\left(1-\prod_{j=1}^N(1-\beta_{ij}(t)I_i(t))\right)
\end{align*}
where $h\in(0,1]$ is the sampling parameter. Using the mean-field approximation, we have
\begin{align}\label{mean-field-1}
    E[I_i(t+1)]=&(1-hE[\delta_i(t)])E[I_i(t)]\\ \nonumber
    &+h(1-E[I_i(t)])\left(1-\prod_{j=1}^N(1-E[\beta_{ij}(t)]E[I_j(t)])\right).
\end{align}
Let the quantities $x_i(t), B_{ij}$ and $D_i$ correspond respectively to $E[I_i(t)]$, $E[\beta_{ij}(t)]$ and $E[\delta_i(t)]$ for all $(i,j)\in V$. Applying the approximation $(1-a)(1-b)\approx(1-a-b)$ for $a\ll 1, b\ll 1$, transforms $$\prod_{j=1}^N(1-E[\beta_{ij}(t)]E[I_j(t)])\approx1-\sum_{j=1}^NE[\beta_{ij}(t)]E[I_j(t)]$$ to yield the approximation for the discrete-time SIS model. Consequently, the equation \eqref{mean-field-1} can be reduced to the following form
\begin{align}\label{D-SIS-origin}
    &x_i(t+1)=(1-hD_i)x_i(t)+h(1-x_i(t))\sum_{j=1}^N B_{ij}x_j(t)
\end{align}

\begin{assumption}\label{assp2}
    $hD_i<1$ and $h\sum_{i\ne j}B_{ij}<1$, $\forall i\in[N]$.
\end{assumption}

Many authors researched the discrete SIS model \eqref{D-SIS}, for instance \cite{Pare2020, Gracy2021}. To represent the SIS model in matrix form, we define a column vector $\textbf{x}(t)$ with elements $x_i(t)$, a matrix $\textbf{B}$ with elements $B_{ij}$, and matrices $\textbf{D}$ and $\textbf{X}(t)$ as the diagonal matrices $diag(D_i)$ and $diag(x_i(t))$, respectively. Hence, \eqref{D-SIS-origin} becomes
\begin{align}\label{D-SIS}
    &\textbf{x}(t+1)=(I-h\textbf{D})\textbf{x}(t)+h(I-\textbf{X}(t))\textbf{B}\textbf{x}(t).
\end{align}
The model \eqref{D-SIS} is called the Deterministic SIS model (D-SIS for short).

In this paper, we study the cumulative number of infections $\textbf{x}(t)$ and the deviation in the number of infections. We denote $\textbf{x}^*$ with entries $x^*_i := \sup\limits_t (x_i(t))$. The following definition indicates the deviation of the infection prevalence from the initial infection prevalence on the network, see also \cite{Yi2022}. 
\begin{definition}
We say that the quantity $\sigma(P) := \|\textbf{x}^* - \textbf{x}(0)\|_1$ is the {\em increased number of infections} of the network.
\end{definition}
To calculate the accumulated variation of the infection amount, we define a new quantity as follows.
\begin{definition}
The quantity
\begin{equation}\label{cumulative}
    \sigma'(P) := \sum_{l=0}^{t-1} \|\mathbf{x}(l) - \mathbf{x}(l+1)\|_1
\end{equation}
is called the {\em cumulative infection variation} up to time $t$.
\end{definition}
\begin{problem}\label{problem}
    Given an undirected graph $G=(V,E)$, where each edge $(i,j)\in E$ has an infection rate $B_{ij}$ and each node $i\in V$ has a healing rate $D_i$, an initial state vector $\textbf{x}(0)\in [0,1]^N$, a candidate set $Q\subseteq E$ of size $q$ and a positive integer $k$ satisfying $0 < k \le q$. Find a subset $P^* \subseteq Q$ with $|P^*| \le k$ such that
    $$P^* \in \operatorname{argmin}_{P\subseteq Q, |P|=k} \sigma(P).$$
\end{problem}
\begin{remark} Finding an optimal solution for $P^*$ in Problem \ref{problem} is NP-hard. The proof is similar to \cite[Theorem 4.4]{Yi2022}.
\end{remark}
\subsection{Convergence analysis}
The following theorem gives a sufficient condition for the convergence of the model. Our result is different from \cite[Theorem 1]{Pare2020} because we use $B_{ij} = E[\beta_{ij}(t)]$ instead of $B_{ij} = \beta_ia_{ij}$ in \cite{Pare2020}. We need an assumption about the structure of the matrix $B$. A square matrix is called irreducible if it cannot be permuted to a block upper triangular matrix. 
\begin{assumption}\label{assp3}
    The matrix $\textbf{B}$ is irreducible.
\end{assumption}

From Equation (\ref{D-SIS}) and Assumption \ref{assp2} we obtain the inequality
\begin{equation}\label{ineq1.5}
    \textbf{x}(t+1)\leq(\textbf{I}-h\textbf{D}+h\textbf{B})\textbf{x}(t).
\end{equation}
Let $\mathbf{M}:= \mathbf{I} - h\mathbf{D} + h\mathbf{B}$ be the state transition matrix of the D-SIS. Then, we have the following important inequality.
\begin{align}\label{ineq1}
    \textbf{x}(t+1)\leq \textbf{M}\textbf{x}(t).
\end{align}
The following result is similar to \cite[Theorem 1]{Pare2020}, in here, by using the above inequality, we give another proof for it, see the Appendix.
\begin{theorem}\label{convergence}
    Suppose that the Assumptions \ref{assp1}, \ref{assp2}, \ref{assp3} hold for the D-SIS model \eqref{D-SIS}. If the spectral radius of $\textbf{M}$ satisfies $\rho(\textbf{M})<1$, then the network is asymptotically stable with the equilibrium point $\hat{\textbf{x}}=0$.
\end{theorem}
\section{Supermodularity for D-SIS model and the Edge Deletion Algorithm}\label{sec3}
Let $\mathbf{M}_{-P}:= \mathbf{I} - h\mathbf{D} + h\mathbf{B}_{-P}$ be the state transition matrix of the D-SIS model after removing the edge set $P \subseteq Q$. Then, we have the following important inequality.
\begin{align*}
    \textbf{x}(t+1)\leq(\textbf{M}_{-P})\textbf{x}(t).
\end{align*}
The following result aims to upper bound the quantities $\sigma(\cdot)$ and $\sigma'(\cdot)$ to avoid solving an NP-hard problem. The upper-bounded quantity will be supermodular and monotone, which allows the problem to be solvable in polynomial time.
\begin{theorem}\label{upper bound}
    Suppose that the Assumptions \ref{assp1}, \ref{assp2}, \ref{assp3} hold for the D-SIS model \eqref{D-SIS}. If $\rho(\textbf{M}_{-P})<1$, then the increased number of infections of the network satisfies  
    \begin{align}\label{Sigma mu}
    \sigma(P)\leq\sigma'(P)\leq\hat\sigma(P):=1^T(\textbf{M}_{-P}+h\textbf{D}-\textbf{I})(\textbf{I}-\textbf{M}_{-P})^{-1}\textbf{x}(0),
    \end{align}
    where $\sigma'(P)$ is defined by \eqref{cumulative}.
\end{theorem}

\begin{proposition}\label{supermodularity}
    The function $\hat\sigma(P)$ is non-increasing and supermodular with respect to the edge removal set $P \subseteq Q$.
\end{proposition}
It is computationally intractable to design an algorithm that solves all instances in polynomial time. Therefore, we employ a monotone supermodular upper bound $\hat\sigma(.)$ of the objective function as a surrogate for optimization.
\begin{algorithm}[H]
\caption{Greedy Algorithm (Edge Deletion Algorithm)}\label{GA}
    \textbf{Input:} A function $\hat\sigma,$ a network $G,$ initial states of nodes, a candidate set $Q,$ an integer $k$.
    
    \textbf{Output:} A set of $k$ edges $P\subseteq Q$.
    
    Initialize the set $P = \emptyset$;
    
    \For{$i = 1$ to $k$}{
        Compute $\hat\sigma(P\cup\{e\})$ for each $e\in Q\setminus P$;
         
        $e^* \leftarrow{} \arg\max_{e \in Q \setminus P} \left(\hat\sigma(P) - \hat\sigma(P \cup \{e\}) \right)$;
        
        $P \leftarrow{} P \cup \{e^*\}$ 
        }
    \Return $P$;
\end{algorithm}

\section{Experimental results}\label{sec4}
\subsection{Comparison with heuristic algorithms}
To evaluate the effectiveness of the proposed method, we conduct an experimental comparison with three baseline algorithms, where the sampling parameter for both models is set to $h=0.9$.
\begin{itemize}
\item \textbf{Random} \cite{Callaway2000}: Randomly removes an edge from the network.
\item \textbf{Max-degree} \cite{Albert2000}: Removes an edge adjacent to the node with the highest degree in each round.
\item \textbf{Greedy Algorithm (EDA)} \cite{Yi2022}: Removes an edge in each round based on the objective function $\hat\sigma(.)$ in \eqref{Sigma mu}, as shown in Algorithm \ref{GA}.
\end{itemize}

\begin{figure}[H]
    \centering
    \begin{minipage}[t]{0.48\textwidth}
        \centering
        \includegraphics[width=\linewidth]{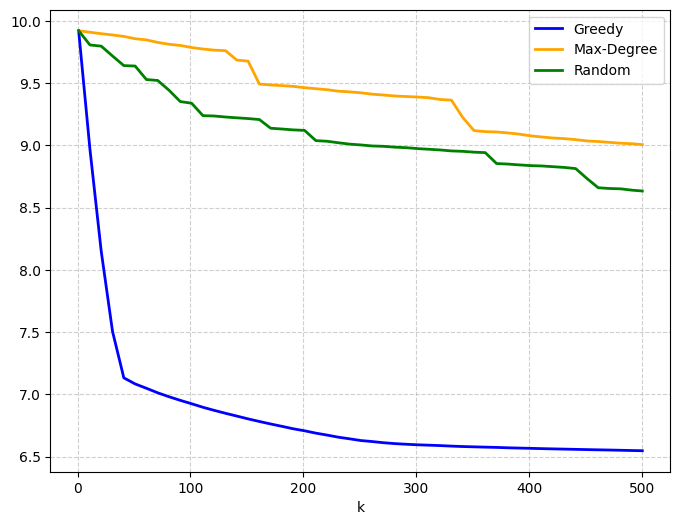}
        {(The Erdős–Rényi Network)}
        \label{er}
    \end{minipage}
    \hfill
    \begin{minipage}[t]{0.48\textwidth}
        \centering
        \includegraphics[width=\linewidth]{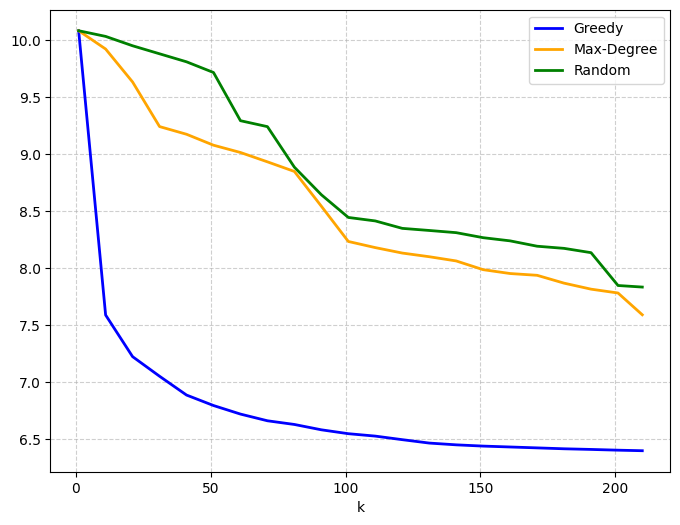}
        {(The Haslemere Network)}
        \label{hasle}
    \end{minipage}
    \caption{The number of infections after deleting edges on the ER and Haslemere networks}
    \label{exp-result}
\end{figure}

We performed simulations on an Erdős-Rényi model (ER for short) consisting of an undirected graph with $700$ nodes and a connection probability of $p = 0.023$. The generated graph $G=(V,E)$ with $5568$ edges. The healing rate for each node $i$ is randomly selected from the interval $D_i \in [0.5, 0.77]$, while the infection rate for each edge $(i, j) \in E$ is randomly distributed within $B_{ij} \in [0.017, 0.039]$. To initiate the propagation process, a seed set $S$ of $5$ nodes ($|S| = 5, S \subset V$) is randomly selected with initial infection states $x_u(0) \in [0.8, 0.94], \forall u \in S$. Conversely, the initial infection state for all remaining nodes is set to $x_u(0) = 0, \forall u \notin S$. To optimize the network structure for Problem \ref{problem}, a candidate edge set $Q$ is constructed by randomly selecting $|Q| = 2,500$ edges from the original edge set, with the objective of removing $k = 500$ edges.

In the Haslemere network model collected in England using BBC Pandemic Haslemere app \cite{Klepac2018, Firth2020}, we restrict the influence of high-degree nodes to ensure a more realistic setting. A data preprocessing step was conducted to reduce the maximum node degree from $37$ to a threshold of at most $15$. Following this reduction, the network retains $1162$ edges. Subsequently, the dynamical parameters are configured similarly, with the node recovery rate set to $D_i \in [0.3, 0.58], \forall i \in V$ and the edge infection rate distributed within $B_{ij} \in [0.04, 0.049], \forall (i, j) \in E$. We also select a seed set of 5 nodes, where the initial state satisfies $x_u(0) \in [0.87, 0.98]$ for all $u \in S$. Finally, the candidate set $Q$ is formed by randomly selecting $520$ edges from the preprocessed edge set to evaluate the removal of $k = 210$ edges.

The simulation results presented in Figure \ref{exp-result} indicate that the experimental results derived from the upper bound function under the Edge Deletion Algorithm achieve the highest edge-cutting efficiency on both the ER and Haslemere contact networks, outperforming the other two algorithms. 

\setlength{\tabcolsep}{7pt} 
\begin{table}[htbp]
  \centering
  \caption{Performance Comparison of Algorithms in the Haslemere Network}
  \label{tab:sample_table}
  \renewcommand{\arraystretch}{1.3}
\begin{threeparttable}
  \begin{tabular}{c l l l l l l l l l}
    \hline
    \textbf{Test} & \textbf{Initial Infections} &
    \multicolumn{2}{c}{\textbf{Random}} &
    \multicolumn{2}{c}{\textbf{Max-Degree}} &
    \multicolumn{2}{c}{\textbf{EDA}} \\ 
    \cline{3-8}
    & & RT & RI & RT & RI & RT & RI \\ 
    \hline
    1 & 8.821 & 75.436	 & 1.619 & 61.596 & 2.027 & 49.636 & 2.721 \\
    2 & 9.590 & 54.084 & 1.632 & 45.905 & 1.404 & 41.357 & 2.752 \\
    3 & 9.897 & 97.085 & 1.528 & 64.147 & 2.022 & 50.512	 & 3.579 \\
    4 & 13.552 & 89.250	 & 2.873 & 57.307  & 4.342 & 50.406 & 5.946 \\
    5 & 7.698 & 50.607 & 1.414 & 47.496 & 1.243 & 39.781 & 2.336 \\
    6 & 9.872 & 57.553 & 2.194 & 47.935 & 2.062 & 42.863 & 3.359 \\
    7 & 10.558 & 56.259 & 2.052 & 46.298 & 2.156 & 40.098 & 3.663 \\
    8 & 8.872 & 55.338 & 1.468 & 45.423 & 1.993 & 44.497 & 2.946 \\
    9 & 12.176 & 62.565 & 2.256 & 51.615 & 3.364 & 47.107 & 4.680 \\
    10 & 8.580 & 67.979 & 1.619 & 49.769 & 2.533 & 43.900 & 3.177 \\
    \hline
   & \textbf{Standard Deviation} &  15.870 & 0.467 & 6.848 & 0.919 & 4.165 & 1.073\\
  \hline
  \end{tabular}
  \begin{tablenotes}
      \small
      \item RT - Run Time (s)
      \item RI - Reduction in Infections
    \end{tablenotes}
    \end{threeparttable}
\end{table}

Table \ref{tab:sample_table} presents a sample-based performance comparison of three edge-removal strategies - Random, Max-Degree, and EDA - on the Haslemere network. These methods are directly evaluated based on two core metrics: execution time (RT) and the reduction in infections (RI) relative to the initial infection count.

\subsection{Convergence Analysis}
\begin{figure}[H]
    \centering
    \begin{minipage}[t]{0.48\textwidth}
        \centering
        \includegraphics[width=\linewidth]{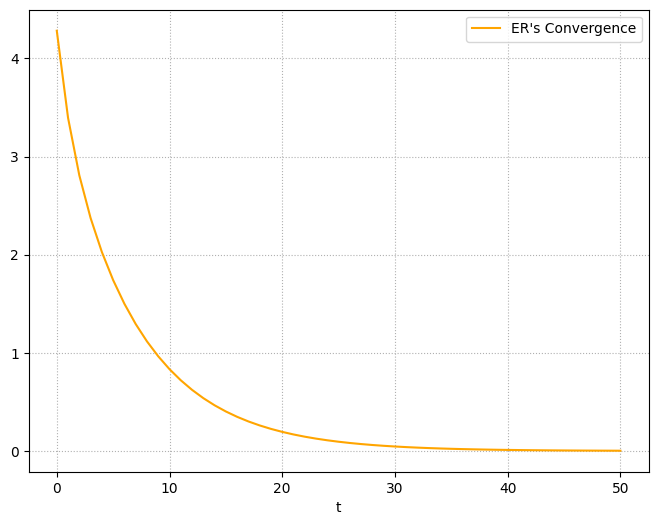}
        {(The Erdős–Rényi Network)}
        \label{er}
    \end{minipage}
    \hfill
    \begin{minipage}[t]{0.49\textwidth}
        \centering
        \includegraphics[width=\linewidth]{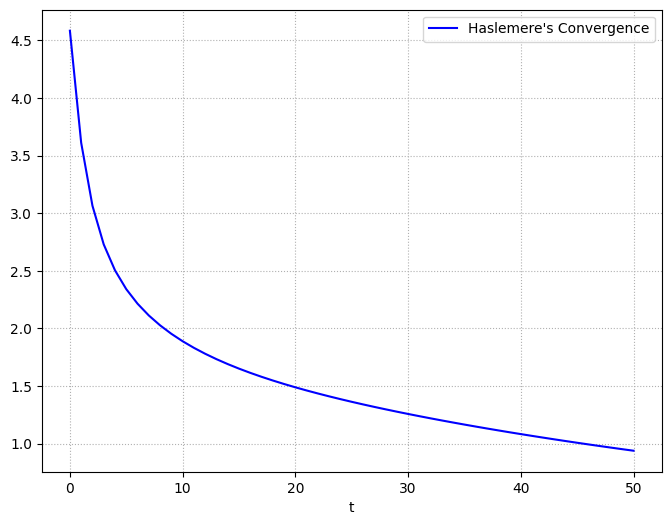}
        {(The Haslemere Network)}
        \label{hasle}
    \end{minipage}
    \caption{Convergence of infection states over time on The ER and Haslemere networks}
    \label{convergence_exp}
\end{figure}

We examine the convergence of infection states in the ER and Haslemere networks, as shown in Figure \ref{convergence_exp}. The results demonstrate that even without edge removal, the infection state consistently approaches the equilibrium $\hat{\mathbf{x}} = \mathbf{0}$ as $t \to \infty$, provided that the spectral norm of the state transition matrix for both models is strictly less than $1$.
\section{Conclusion}\label{conclusion}
We studied the deterministic SIS model defined on discrete-time, focusing on optimizing infections in complex networks under the proposed conditions. By employing the condition that the spectral norm of the system matrix is less than $1$, the system converges to the healthy state. Furthermore, we demonstrated that under these conditions, the cumulative number of infections is upper-bounded by a monotone function. For optimal mitigation, the Edge Deletion Algorithm (Algorithm \ref{GA}) strategically removes links from selected candidate edges. To validate our theoretical findings, we evaluated the accuracy and effectiveness of the proposed algorithm through numerical simulations on both random graphs and real-world datasets. 

Future work can consider extending the SIS model to various graph topologies. Additionally, we would like to discover new algorithms capable of tracking more diverse and dynamic network propagation processes.
\section*{Appendix}
\begin{proof}{\bf\em of Theorem \ref{convergence}.}
Since the assumption, the spectral norm of the state transition matrix $\textbf{M}$ is less than $1$. It follows from Inequality \eqref{ineq1} that $\|\textbf{x}(t)\| \leq \|\textbf{M}\|\|\textbf{x}(t-1)\|$ for any time step $t$. By iterating this inequality over $t$ time steps, we obtain $$\|\textbf{x}(t)\| \leq (\rho(\textbf{M}))^t\|\textbf{x}(0)\|.$$ Consequently, $\|\textbf{x}(t)\|$ converges to $\hat{\textbf{x}} = \mathbf{0}$ as $t \to \infty$.
\end{proof}

\begin{proof}{\bf\em of Theorem \ref{upper bound}.}
By induction on the time steps for the Inequality \eqref{ineq1}, we obtain
    \begin{align*}
        \textbf{x}(t-1)&\leq \textbf{M}^{t-1}\textbf{x}(0) \\
        \textbf{x}(t-2)&\leq \textbf{M}^{t-2}\textbf{x}(0) \\
        & \dots \\
        \textbf{x}(1)&\leq \textbf{M}\textbf{x}(0).
    \end{align*}
    It follows that
    \begin{align}\label{proof2.3}
        \sum\limits_{l=0}^{t-1}\textbf{x}(l) \leq \left((\textbf{M}_{-P})^{t-1}+\dots+\textbf{M}_{-P}+\textbf{I}\right)\textbf{x}(0).
    \end{align}
Since the assumption $\rho(\textbf{M}) < 1$, we have, for instance \cite{Horn2012},
\begin{align*}
    \left((\textbf{M}_{-P})^{t-1}+\dots+\textbf{M}_{-P}+\textbf{I}\right)\textbf{x}(0) \approx (\textbf{I}-\textbf{M}_{-P})^{-1}\textbf{x}(0).
\end{align*} Furthermore, \eqref{ineq1.5} implies that $\textbf{x}(t+1) \leq \textbf{x}(t)+h\textbf{B}_{-P}\textbf{x}(t)$, which indicates 
\begin{align*}
    \textbf{x} (t) - \textbf{x} (t-1)&\leq h\textbf{B}_{-P}\textbf{x}(t-1)  \\
    \textbf{x} (t-1) - \textbf{x} (t-2)&\leq h\textbf{B}_{-P}\textbf{x}(t-2)  \\
    & \dots \\
    \textbf{x}(1) - \textbf{x}(0)&\leq h\textbf{B}_{-P}\textbf{x}(0).
\end{align*}
Hence,
\begin{align}\label{proof2.1}
    \sum_{l=0}^{t-1}\|\textbf{x}(l)-x(l+1)\|_1&\leq\sum_{k=0}^{t-1}\|h\textbf{B}_{-P}\textbf{x}(k)\|_1.
\end{align}

Since the Assumption \ref{assp1}, we have $\textbf{x}(t)\geq0$ and $(\textbf{B}_{-P})_{ij}\geq0,\forall t\geq0,\forall (i,j)\in E$, therefore 
\begin{align}\label{proof2.2}
    \sum_{k=0}^{t-1}\|h\textbf{B}_{-P}\textbf{x}(k)\|_1 = \|h\textbf{B}_{-P}\sum_{k=0}^{t-1}\textbf{x}(k)\|_1
\end{align}
Combining \eqref{proof2.3}, \eqref{proof2.1} and \eqref{proof2.2}, we obtain that the cumulative number of infections at any time step $t$ is given by
\begin{align*}
    \|\textbf{x}(t)-\textbf{x}(0)\|_1 &\leq \sum_{l=0}^{t-1}\|\textbf{x}(l)-\textbf{x}(l+1)\|_1 \\
    &\leq \|h\textbf{B}_{-P}\sum_{k=0}^{t-1}\textbf{x}(k)\|_1 \\
    &\leq\|(\textbf{M}_{-P}+h\textbf{D}-\textbf{I})(\textbf{I}-\textbf{M}_{-P})^{-1}\textbf{x}(0)\|_1=\hat\sigma(P).
\end{align*}
The theorem is proved.
\end{proof}
\begin{proof}{\bf\em of Proposition \ref{supermodularity}.}
Let $P_1 \subseteq P_2 \subseteq Q$ be the sets of edges to be removed from the graph $G=(V,E)$ such that $\rho(\textbf{M}_{-P_1}) < 1$ and $\rho(\textbf{M}_{-P_2}) < 1$. Then, for any $e \in Q \setminus P_2$, it holds that 
$$(\textbf{I}-\textbf{M}_{-P_1})^{-1}-(\textbf{I}-\textbf{M}_{-P_1\cup\{e\}})^{-1}\geq(\textbf{I}-\textbf{M}_{-P_2})^{-1}-(\textbf{I}-\textbf{M}_{-P_2\cup\{e\}})^{-1}.$$

Using the Sherman-Morrison formula \cite{Bartlett1951}, we have
\begin{align*}
    (\textbf{I}-\textbf{M}_{-P} + ce_ie_j^T)^{-1} = (\textbf{I}-\textbf{M}_{-P})^{-1}-c\frac{(\textbf{I}-\textbf{M}_{-P})^{-1}e_ie_j^T(\textbf{I}-\textbf{M}_{-P})^{-1}}{1+ce_j^T(\textbf{I}-\textbf{M}_{-P})^{-1}e_i}
\end{align*}
in which $c=hB_{ij}$ is a nonnegative number. Take $(\textbf{I}-\textbf{M}_{-P_1\cup\{e\}})^{-1}=(\textbf{I}-\textbf{M}_{-P_1} + ce_ie_j^T)^{-1}$,
we transform into 
\begin{align*}
    (\textbf{I}-\textbf{M}_{-P_1})^{-1}-(\textbf{I}-\textbf{M}_{-P_1\cup\{e\}})^{-1}
    &=c\frac{(\textbf{I}-\textbf{M}_{-P_1})^{-1}e_ie_j^T(\textbf{I}-\textbf{M}_{-P_1})^{-1}}{1+ce_j^T(\textbf{I}-\textbf{M}_{-P_1})^{-1}e_i} \\
    &=\int_{0}^{1} c(\textbf{I}-\textbf{M}_{-P_1})^{-1}e_ie_j^T(\textbf{I}-\textbf{M'}_{-P_1})^{-1}t\,dt \\
    &\geq \int_{0}^{1} c(\textbf{I}-\textbf{M}_{-P_2})^{-1}e_ie_j^T(\textbf{I}-\textbf{M}_{-P_2})^{-1}t\,dt \\
    &=(\textbf{I}-\textbf{M}_{-P_2})^{-1}-(\textbf{I}-\textbf{M}_{-P_2\cup\{e\}})^{-1}
\end{align*}
Thus, we obtain
\begin{align*}
    (\textbf{I}-\textbf{M}_{-P_1})^{-1}-(\textbf{I}-\textbf{M}_{-P_1\cup\{e\}})^{-1}\geq(\textbf{I}-\textbf{M}_{-P_2})^{-1}-(\textbf{I}-\textbf{M}_{-P_2\cup\{e\}})^{-1}, e\in Q\setminus P_2.
\end{align*}
In addition, $\textbf{M}_{-P}+h\textbf{D}-\textbf{I}$ is an entrywise nonnegative and nonincreasing supermodular function of $P$. Therefore, the product of $(\mathbf{M}_{-P} + h\mathbf{D} - \mathbf{I})$ and $(\textbf{I}-\textbf{M}_{-P})^{-1}$ is entrywise monotone supermodular. It follows that the function $\hat{\sigma}(.)$ is non-increasing and supermodular.
\end{proof}

\end{document}